\documentclass[11pt,a4paper]{article} 
\usepackage{natbib}
\usepackage[hyperindex, hidelinks]{hyperref}
\usepackage{stmaryrd}
\usepackage{graphicx}
\usepackage{wrapfig}
\usepackage{subcaption}
\usepackage{amssymb}
\usepackage{bbm}
\usepackage{verbatim}
\usepackage{enumerate}
\usepackage[hyperindex]{hyperref}
\usepackage{mathtools} 
\usepackage{amsmath, amsfonts, amsthm, amscd}
\usepackage{fancybox}
\usepackage[T1]{fontenc}
\newlength{\querylen}
\usepackage{tikz}
\usepackage{tikz-qtree}
\usetikzlibrary{bayesnet}
\usepackage{mathtools}

\newcommand{\R}{\mathbb{R}}

\newcommand{\E}{\mathbb{E}}

\theoremstyle{plain}

\newtheorem{theorem}{Theorem}

\newtheorem{assumption}{Assumption}

\newtheorem{lemma}{Lemma}

\newtheorem{?}{Question}

\newtheorem{algorithm}{Algorithm}

\usepackage{titlesec}
\titleformat{\subsubsection}[runin]
{\normalfont\itshape}{\thesubsubsection}{1em}{}

\def \E {\mathbb{E}}

\def\simiid{\overset{\text{i.i.d.}}{\sim}}
\def \R {\mathbb{R}}

\title{Multiplicative Errors-in-Variables Models for Hyperspectral Unmixing}
\author{Vivek Singh and Peter Hoff \\\\
Department of Statistical Science, Duke University}
\date{\today}

\begin{document}
\maketitle

\begin{abstract}
Estimating endmember spectra and their corresponding abundances from hyperspectral images is a fundamental inverse problem in remote sensing. The standard linear mixing model relies on additive homoscedastic errors, which misspecifies the mean-variance relationship inherent to optical measurements for which the noise intensity might be proportional to the signal magnitude. This leads to confidence intervals with poor coverage. To address these limitations, we propose a hierarchical Bayesian multiplicative Errors-in-Variables (EIV) model. Motivated by the physics of light propagation, our formulation captures signal-dependent noise through a multivariate log-normal specification that naturally accommodates dependence across spectral bands. The EIV framework treats endmember signatures as random and allows for class-specific variance scaling, yielding a flexible and physically grounded data-generating process. We establish that the MLE of the abundance vector under the multiplicative model has a constant asymptotic multivariate coefficient of variation, implying that well-calibrated confidence intervals should scale proportionally with signal magnitude. The additive model, by contrast, produces intervals of constant width regardless of signal magnitude, leading to overcoverage at low abundance and undercoverage at high abundance when the true noise is multiplicative. Practical implementation and inference for the full hierarchical model is obtained from posterior distributions over the model parameters. Experiments on several real and simulated datasets demonstrate that confidence intervals derived from the proposed model achieve improved coverage across abundance levels compared to those from a standard additive model, with widths that scale proportionally with signal magnitude, while achieving comparable or superior signal reconstruction error.

\end{abstract}

\section{Introduction}
\label{sec:intro}
Hyperspectral imaging serves as a robust analytical framework for the spectroscopic characterization of materials within an observed surface area of interest, or scene 
\cite{BioucasDias2013HyperspectralRS}. However, due to inherent sensor resolution constraints and the spatial heterogeneity of the Earth's surface, individual pixels often represent a mixture of multiple distinct materials \cite{SU}. This phenomenon necessitates spectral unmixing, a source separation task aimed at identifying the constituent substances (endmembers) and quantifying their respective spatial proportions (abundances) within the mixed observations. 
The linear mixing model is the predominant paradigm in spectral unmixing literature \cite{SU, LMM}. It postulates that the observation vector $y \in \R_+^n$, where $n$ is the number of bands, is a linear combination of the scene's $p$ endmembers perturbed by additive noise:
\begin{equation}
    y = Mb + \epsilon,
\end{equation}
where $b \in \mathbb{R}^p_+$ represents the abundance vector subject to non-negativity constraints, $M \in\R_+^{n\times p}$ represents the matrix of endmember spectra, and $\epsilon \in\R^n$ denotes the noise term. Standard formulations assume $\epsilon$ follows a zero mean Gaussian distribution. Non-negative matrix factorization (NMF) approaches that estimate the abundance vector $b$ also implicitly rely on this additive assumption by minimizing the residual sum of squared error, which is equivalent to maximum likelihood estimation under a Gaussian noise model \cite{NMF_LS}.

However, beyond point estimation, a critical question is how reliable these estimates are. For uncertainty quantification, what matters is that the noise model correctly captures the variance of the observations, since confidence intervals depend on the estimated variance. Under the additive formulation, the variance of $y$ is assumed to be constant across different levels of the signal, which is at odds with the physics of optical measurements where noise intensity is multiplicative, scaling with signal magnitude due to effects such as atmospheric path radiance and topographic shading \cite{Richter01012002}. Indeed, signal-dependent multiplicative noise in hyperspectral imagery has been acknowledged in prior work. \cite{IS_div} minimizes the Itakura-Saito divergence, equivalent to maximum likelihood estimation under a multiplicative Gamma noise model, and \cite{Acito2011} explicitly models signal-dependent noise parameters. However, both approaches assume independence across spectral bands and focus exclusively on point estimation, leaving uncertainty quantification unaddressed.

The misspecified mean-variance relationship in the additive model leads to a fundamental miscalibration: the additive model produces confidence intervals that are too wide at low abundance values and too narrow at high ones. A correctly specified model must instead yield confidence intervals whose width scales proportionally with the signal, a property we establish theoretically and demonstrate empirically for our proposed multiplicative formulation. Figure~\ref{fig:add_v_mult} illustrates the two key features of the data that motivate our model on the Asphalt endmember from the Urban dataset \cite{zhu2014spectral}. The left panel shows the marginal distribution of reflectance at a single band across reference pixels, which is clearly right-skewed, a signature of signal-dependent multiplicative noise rather than symmetric additive noise. The right panel shows the estimated log-scale covariance matrix across all spectral bands, revealing strong positive dependence across bands that cannot be captured by a diagonal covariance structure.

\begin{figure}[!htbp]
    \centering
    \includegraphics[width=\columnwidth]{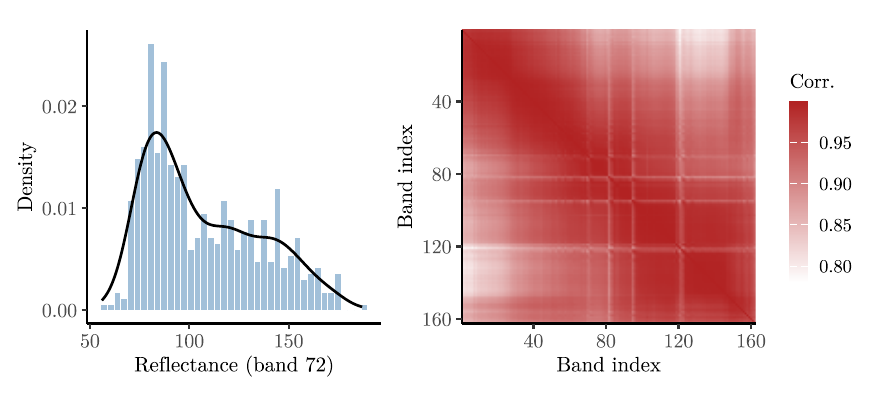}
    \caption{Marginal distribution and inter-band dependence structure of the Asphalt endmember class from the Urban dataset \cite{zhu2014spectral}. Left: histogram and kernel density estimate of the reflectance at a single spectral band across reference pixels, illustrating the right-skewed marginal distribution characteristic of multiplicative noise. Right: estimated log-scale correlation matrix across all spectral bands, showing strong positive dependence across bands that motivates the use of a full covariance matrix $\Sigma$.}
    \label{fig:add_v_mult}
\end{figure}

The standard practive with LMM is to take the endmember matrix $M$ as given. In many applications, however, $M$ is never directly observed. Even pixels known to be composed predominantly of a single material class exhibit significant within-class spectral variability due to intrinsic material heterogeneity, localized degradation, and micro-topography \cite{borsoi2021spectral}. This naturally motivates an Errors-in-Variables (EIV) formulation \cite{Fuller1987}, where the reference pixels are treated as samples from a class-specific distribution rather than exact observations of the true signatures. It is well established that plugging in aggregated estimates of $M$ in place of the true signatures introduces bias in the resulting abundance estimates \cite{Fuller1987}. 

Furthermore, this within-class variability is not uniform across material classes. Complex organic surfaces such as vegetation exhibit higher spectral variance due to differences in chlorophyll content and leaf structure, whereas man-made surfaces like concrete or asphalt tend to be relatively homogeneous \cite{borsoi2021spectral}. A principled model must therefore allow for class-specific variance profiles rather than assuming a single global level of spectral variability across all endmembers.

In this article, we propose a hierarchical approach to hyperspectral unmixing that addresses the issues raised above. We model the data generation process using a multiplicative error model, which is robust in scenarios where signal-dependent noise dominates. Our contributions are as follows:
\begin{itemize}
    \item Uncertainty quantification: We model the observations using a multivariate log-normal distribution, which captures the signal-dependent mean-variance relationship inherent to optical measurements. As established in Section~\ref{sec:asymptotics}, this yields well-calibrated confidence intervals whose width scales proportionally with the signal magnitude, in contrast to the constant-width intervals produced by the additive model.
    \item Within and across-class spectral variability: We adopt a hierarchical EIV formulation where the reference pixels are treated as samples from class-specific distributions, accommodating endmember heterogeneity within each material class. Furthermore, the degree of spectral variability is allowed to differ across classes, reflecting the physical differences between material classes. A shared full covariance matrix $\Sigma$ explicitly captures inter-band dependencies across spectral channels, common to both the mixed pixel and the reference pixels.
\end{itemize}

Both contributions are essential for well-calibrated uncertainty quantification. The scaling of interval widths with signal magnitude is established theoretically in Section~\ref{sec:asymptotics}, while the EIV formulation accounts for the additional uncertainty introduced by the unobserved endmember signatures.

The rest of the article is organized as follows. Section~\ref{sec:asymptotics} considers a simplified setting with known endmember matrix and scalar covariance, and establishes that the MLE of the abundance vector under the multiplicative model has a constant asymptotic multivariate coefficient of variation, suggesting that the interval widths should scale with signal magnitude, and demonstrates via simulation that intervals constructed assuming the additive model are miscalibrated when the true noise is multiplicative. Section~\ref{sec:model} introduces the full hierarchical EIV model with class-specific variance scaling and a shared spectral covariance structure, and describes an MCMC algorithm for posterior inference. Section~\ref{sec:experiments} evaluates both models on real datasets, comparing coverage rates, interval widths, and signal reconstruction error across abundance levels. Section~\ref{sec:discussion} discusses limitations of our multiplicative model and directions for future work. The replication code and an \texttt{R} package implementing the proposed model are available on the first author's webpage. 


\section{Adaptive Confidence Intervals via a Multiplicative Model}
\label{sec:asymptotics}

To understand the effect of the noise model on uncertainty quantification, we consider a simplified unmixing setting where the endmember matrix $M$ is assumed to be known and the covariance is proportional to the identity matrix. Under these assumptions, we compare the asymptotic behavior of confidence intervals for the abundance vector $b$ under the additive and multiplicative formulations.

The additive model assumes the observed pixel is a linear mixture of endmembers perturbed by homoscedastic noise:
\begin{align}
\label{model:add}
    y = Mb + \epsilon, \qquad \epsilon \sim \mathcal{N}_n(0, \sigma^2 I),
\end{align}
where $y \in \mathbb{R}_+^n$ is the observed spectral signature, $M \in \mathbb{R}_+^{n \times p}$ is the endmember matrix, $b \in \mathbb{R}_+^p$ is the unknown abundance vector, and $\sigma^2$ controls the noise level. The key feature of this model is that the noise is independent of the signal: the variance of $y$ is $\sigma^2 I$ regardless of the value of $Mb$.

The multiplicative model assumes the noise scales with the signal, so that the observed spectrum is an element-wise product of the signal and a noise term:
\begin{align}
\label{model:mult}
    y = (Mb) \odot \varepsilon, \qquad \log(\varepsilon) \sim \mathcal{N}_n(0, \sigma^2 I),
\end{align}
where $y \in \mathbb{R}_+^n$ and all other quantities are as above. Here the standard deviation of $y_i$ is proportional to $m_i^T b$, reflecting the mean-variance relationship inherent to optical measurements.

\subsection{Asymptotic Properties}
\label{sec:asymp_properties}

We characterize the MLE under the correctly specified multiplicative model \eqref{model:mult}. Write $m_i^T$ for the $i$-th row of $M$, let $b_0$ denote the true abundance vector, and set $w_i(b) := m_i / (m_i^T b)$ and $w_i := w_i(b_0)$.

Taking logarithms in \eqref{model:mult} turns the model into the nonlinear regression
\begin{equation}\label{eq:logmodel}
    z_i := \log y_i = \log(m_i^T b) + \eta_i, \qquad \eta_i \overset{\text{iid}}{\sim} \mathcal{N}(0, \sigma^2),
\end{equation}
so that the MLE coincides with the nonlinear least squares estimator
\begin{equation}\label{eq:nls}
    \hat b_{\mathrm{MLE}} = \arg\min_{b \in \Theta} \sum_{i=1}^n \bigl(z_i - \log(m_i^T b)\bigr)^2,
\end{equation}
with regression function $f_i(b) = \log(m_i^T b)$, gradient $f_i'(b) = w_i(b)$, and Hessian $f_i''(b) = -w_i(b) w_i(b)^T$. Its large-sample behaviour therefore follows from the classical asymptotic theory of nonlinear least squares \cite{Wu1981}, once the following regularity conditions are met.

\begin{assumption}\label{assump:reg}
\leavevmode
\begin{enumerate}
    \item[(A1)] $b_0 \in \operatorname{int}(\Theta)$ for a compact convex parameter set $\Theta \subset \mathbb{R}^p$.
    \item[(A2)] $\sup_i \|m_i\| < \infty$ and $\inf_i \inf_{b \in \Theta} m_i^T b > 0$. In particular $\sup_{i,\, b \in \Theta} \|w_i(b)\| \le W < \infty$.
    \item[(A3)] $\displaystyle \frac{1}{n} \sum_{i=1}^n \frac{m_i m_i^T}{(m_i^T b_0)^2} \longrightarrow \bar{I}(b_0)$, positive definite.
\end{enumerate}
\end{assumption}
  
\begin{lemma}[Consistency]\label{lemma:consistency}
    Under model \eqref{model:mult} and Assumption \ref{assump:reg}, $\hat{b}_{\mathrm{MLE}} \to b_0$ almost surely as $n \to \infty$.
\end{lemma}
 
\begin{theorem}[Asymptotic Normality]\label{thm1}
    Under model \eqref{model:mult} and Assumption \ref{assump:reg},
    \begin{equation*}
        \sqrt{n}(\hat{b}_{\mathrm{MLE}} - b_0) \xrightarrow{d} \mathcal{N}(0, V(b_0)), \qquad V(b_0) = \sigma^2 \bar{I}(b_0)^{-1}.
    \end{equation*}
\end{theorem}

Under Assumption~\ref{assump:reg}, the consistency and asymptotic normality of $\hat b_{\mathrm{MLE}}$ follow from the nonlinear least squares theory of \cite{Wu1981} applied to \eqref{eq:nls}. The verification of the conditions of \cite{Wu1981} is given in Appendix~\ref{app:wu}.

A notable feature of the limiting variance $V(b)$ is that it possesses a constant multivariate coefficient of variation (MCV) \cite{Aerts2015-ou}, formalized below.
 
\begin{lemma}[Constant Multivariate Coefficient of Variation]\label{lemma1}
    Let $\bar I(b) := \lim_{n \to \infty} \frac1n \sum_{i=1}^n \frac{m_i m_i^T}{(m_i^T b)^2}$. At every $b$ with $m_i^T b > 0$ for all $i$ at which $\bar I(b)$ exists and is nonsingular, the asymptotic variance $V(b) = \sigma^2 \bar I(b)^{-1}$ satisfies
    \begin{equation*}
        b^T V(b)^{-1} b = \sigma^{-2}.
    \end{equation*}
\end{lemma}
 
\begin{proof}
Since $V(b) = \sigma^2 \bar I(b)^{-1}$, we have $V(b)^{-1} = \sigma^{-2} \bar I(b)$, so $b^T V(b)^{-1} b = \sigma^{-2}\, b^T \bar I(b)\, b$. It remains to show $b^T \bar I(b)\, b = 1$. Because $m_i^T b$ is scalar,
\begin{equation*}
    b^T \frac{m_i m_i^T}{(m_i^T b)^2}\, b = \frac{(m_i^T b)^2}{(m_i^T b)^2} = 1 \qquad \text{for every } i,
\end{equation*}
and therefore
\begin{equation*}
    b^T \bar I(b)\, b = \lim_{n \to \infty} \frac1n \sum_{i=1}^n b^T \frac{m_i m_i^T}{(m_i^T b)^2}\, b = \lim_{n \to \infty} \frac1n \sum_{i=1}^n 1 = 1.
\end{equation*}
Hence $b^T V(b)^{-1} b = \sigma^{-2}$.
\end{proof}
 
Lemma~\ref{lemma1} implies that the width of confidence intervals for $b_j$ scales proportionally with the true abundance $b_j$. For example, if $p = 1$, $\sqrt{V(b)} = \sigma\cdot b$, so the asymptotic standard deviation of $\hat{b}_{\text{MLE}}$ is proportional to the true abundance.
 
If instead the multiplicative model is misspecified as additive \eqref{model:add}, the MLE reduces to the ordinary least squares estimator
\begin{equation*}
    \hat{b}_{\text{OLS}} = (M^T M)^{-1} M^T y,
\end{equation*}
whose variance $\sigma^2(M^TM)^{-1}$ does not depend on $b$. Consequently, the resulting confidence intervals have constant width regardless of the true abundance, leading to systematic miscalibration: overestimated interval lengths at low abundance values and insufficient coverage at high abundance values.

\subsection{Numerical Illustration}
\label{sec:numerical_illustration}

We provide numerical validation for Lemma~\ref{lemma1} through simulation. Simulated data are generated according to the simplified multiplicative model \eqref{model:mult} above, with known endmember signatures $M$ and homoscedastic noise $\log(\varepsilon) \sim \mathcal{N}(0, \sigma^2 I)$. Figure~\ref{fig:asymptotics} compares the coverage rates and interval widths of the confidence intervals obtained from the multiplicative and misspecfied additive models for $n = 200$ spectral bands and $p = 4$ endmembers. $\sigma^2$ is set to $0.1$. The additive model \eqref{model:add} produces interval widths that are approximately constant across abundance values, while the multiplicative model yields widths that scale proportionally with the abundance, consistent with the constant MCV property. Consequently, both models achieve correct coverage at low abundance values, though the multiplicative intervals are narrower. At high abundance values, the additive model undercovers while the multiplicative model maintains correct coverage. 
\begin{figure}[!htbp]
    \centering
    \includegraphics[width=.8\columnwidth]{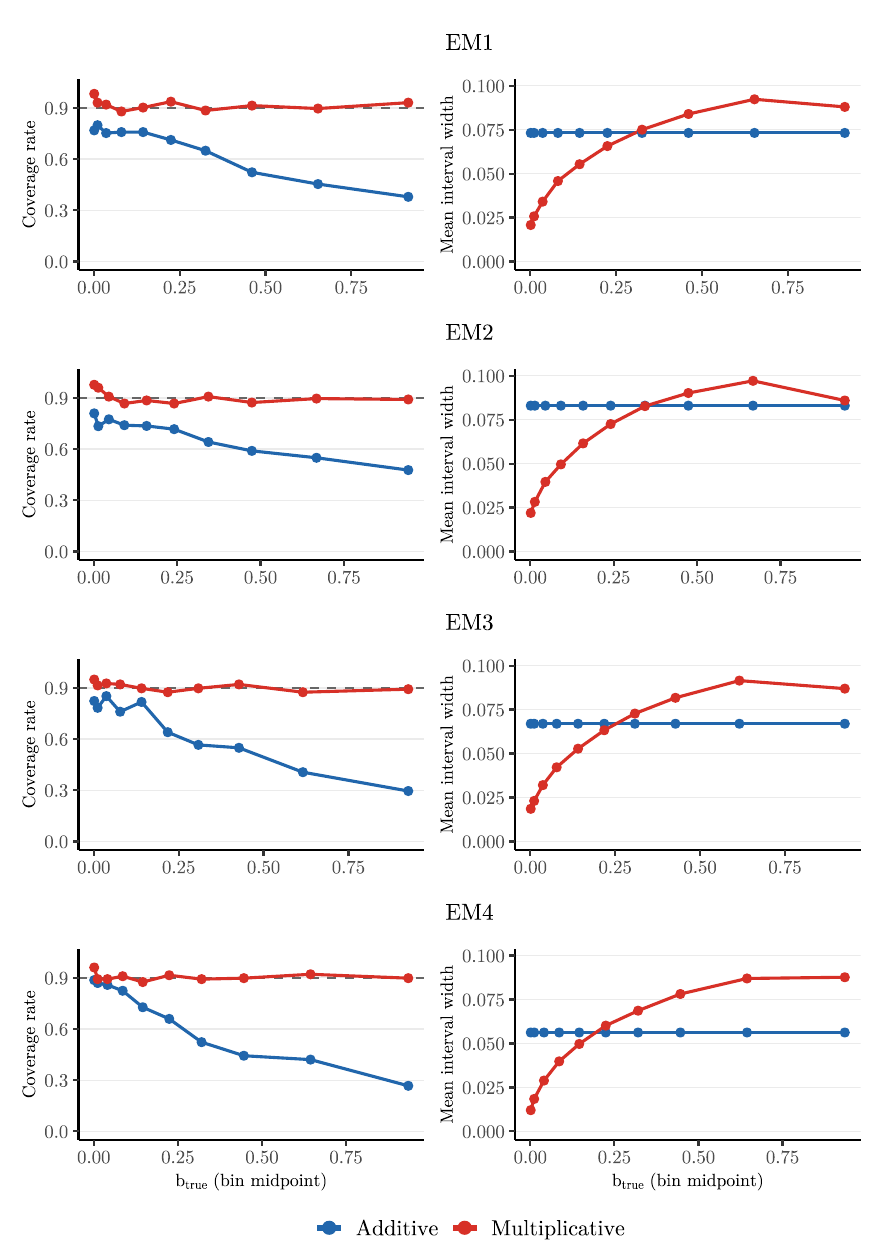}
    \caption{Conditional coverage rates (left) and mean 90\% quantile interval widths (right) as a function of true abundance for each endmember for the multiplicative and misspecified additive fits on simulated data with $n = 200$ spectral bands and $p = 4$ endmembers. Each bin approximately contains $175$ pixels. The dashed line indicates nominal $90\%$ coverage. Blue corresponds to the additive model and red to the multiplicative model.}
    \label{fig:asymptotics}
\end{figure}

\section{Hierarchical Multiplicative Errors-in-Variables Model}
\label{sec:model}

Let $y \in \mathbb{R}_+^n$ denote the observed mixed spectrum across $n$ spectral bands, and let $p$ denote the number of endmember classes in the scene. The goal of unmixing is to recover the abundance vector $b \in \mathbb{R}_+^p$, where $b_j$ quantifies the contribution of the $j$-th material class to the mixed spectrum. Let $\mu_j \in \mathbb{R}_+^n$ denote the median spectrum of the $j$-th endmember class and collect these into the matrix $M = [\mu_1, \dots, \mu_p] \in \mathbb{R}_+^{n \times p}$. We model the mixed spectrum as
\begin{equation*}
    y = Mb \odot \varepsilon,
\end{equation*}
where the multiplicative structure ensures that the noise scales with the signal magnitude. If $M$ were known, this reduces to a non-linear regression problem. In practice, however, $M$ is unobserved and must be estimated from reference spectra. Let $\{x_{ij}\}_{i=1}^{n_j}$ denote the $n_j$ reference spectra associated with the $j$-th endmember class, which are samples from the distribution of that material class. We model them as
\begin{equation*}
    x_{ij} = \mu_j \odot \varepsilon_{ij},
\end{equation*}
where $\varepsilon_{ij}$ captures the within-class spectral variability of the $i$-th reference spectrum. This is precisely the Errors-in-Variables (EIV) setting of \cite{Fuller1987}. It is well established in the EIV literature that naive substitution of an estimate of $M$, such as the sample median of the reference spectra, in place of the true $\mu_j$ leads to biased abundance estimates \cite{Fuller1987}. Our model instead treats $M$ as unknown and infers it jointly with $b$ from the data.

\subsection{Within and Across Endmember Variability}
\label{sec:intraclass}

The multiplicative noise terms $\varepsilon$ and $\varepsilon_{ij}$ are modeled using a multivariate log-normal distribution, which simultaneously captures three sources of variability in the data. First, the noise intensity scales with the signal magnitude, reflecting the mean-variance relationship inherent to optical measurements. Second, spectral bands are dependent and a full covariance matrix $\Sigma \in \mathcal{S}_+^n$ captures these inter-band dependencies. Third, the degree of within-class spectral variability differs across material classes.

To capture these sources jointly, we share the covariance matrix $\Sigma$ across the entire scene while allowing class-specific scaling. Since the reference spectra and the mixed spectrum are drawn from the same scene, they are subject to the same atmospheric and environmental conditions, justifying a common $\Sigma$. Let $\tau_y$ represent the variance scaling for the mixed spectrum and $\tau_j$ the variance scaling for the $j$-th endmember class. The complete model is
\begin{align*}
    y &= Mb \odot \varepsilon, & \log(\varepsilon) \mid \tau_y, \Sigma &\sim \mathcal{N}(0, \tau_y\Sigma), \\
    x_{ij} &= \mu_j \odot \varepsilon_{ij}, & \log(\varepsilon_{ij}) \mid \tau_j, \Sigma &\sim \mathcal{N}(0, \tau_j\Sigma), \quad i = 1, \dots, n_j.
\end{align*}
Figure~\ref{fig:dag} shows a graphical representation of the model. This formulation allows the model to learn a spectral dependence structure ($\Sigma$) shared across the scene, while the scaling parameters $\tau_j$ adaptively quantify the intrinsic variability of each material class. 

\cite{WolbersStahel2005} also employ a log-normal structure in a linear unmixing context but performs point estimation assuming a diagonal error covariance, ignoring inter-band dependencies. \cite{IS_div, Acito2011} acknowledge signal-dependent multiplicative noise but likewise focus on point estimation and do not model inter-band dependence. Our formulation addresses all three sources of variability within a unified hierarchical framework that supports full uncertainty quantification.

\begin{figure}[!htbp]
    \centering
    \begin{tikzpicture}[
        node distance=1.5cm,
        latent/.style={draw, circle, inner sep=2pt, minimum size=1cm},
        obs/.style={draw, circle, fill=gray!20, inner sep=2pt, minimum size=1cm},
        edge/.style={->, >=stealth, thick}
    ]
        \node[latent] (Sigma) at (0, 0) {$\Sigma$};
        \node[latent] (tauy) at (-2, -2) {$\tau_y$};
        \node[latent] (b)    at (-4, -2) {$b$};
        \node[obs]    (y)    at (-2, -4.5) {$y$};
        \node[latent] (tauj) at (2, -2) {$\tau_j$};
        \node[latent] (mu)   at (4, -2) {$\mu_j$};
        \node[obs]    (x)    at (2, -4.5) {$x_{ij}$};
        \draw[edge] (Sigma) -- (y);
        \draw[edge] (Sigma) -- (x);
        \draw[edge] (tauy)  -- (y);
        \draw[edge] (b)     -- (y);
        \draw[edge] (mu)    -- (y);
        \draw[edge] (mu)    -- (x);
        \draw[edge] (tauj)  -- (x);
        \draw[thick] (1.2, -5.2) rectangle (2.8, -3.8) node[below right] {$i = 1:n_j$};
        \draw[thick] (1.0, -5.5) rectangle (5.0, -1.2) node[below right] {$j = 1:p$};
    \end{tikzpicture}
    \caption{Graphical representation of the proposed model. Shaded nodes denote observed quantities. The shared covariance $\Sigma$ governs inter-band dependence across the scene, while $\tau_y$ and $\tau_j$ control the variance scaling for the mixed pixel and the endmember classes respectively.}
    \label{fig:dag}
\end{figure}
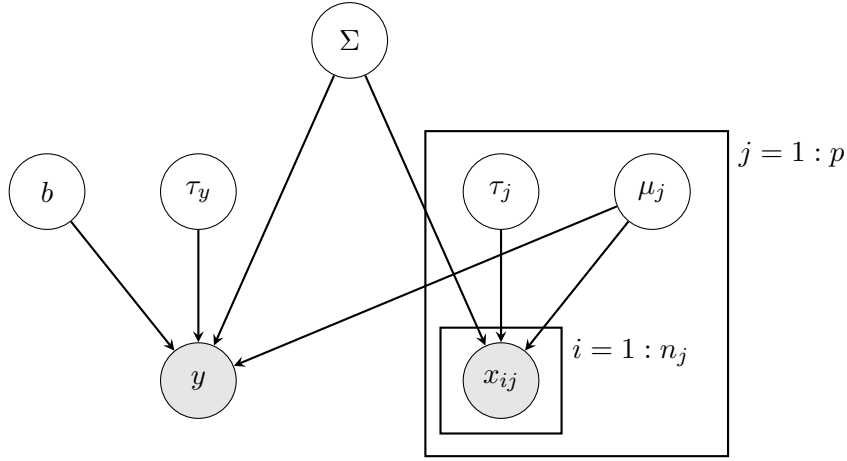

\subsection{Bayesian Inference and Priors}
\label{sec:inference}

We employ a Bayesian framework to jointly infer the abundances $b$, the median endmember signatures $M$, and the variance parameters. The complexity of the multiplicative model precludes a closed-form for the posterior distribution, therefore, we utilize the Metropolis-Hastings algorithm to generate approximate samples.

\subsubsection{Prior Specification.}
\label{sec:priors}
We assign the following priors to the parameters:
\begin{itemize}
    \item \textbf{Abundances.} We employ an Exponential prior on the abundance coefficients to enforce non-negativity: $b_j \simiid \text{Exp}(p)$, for $j=1,\dots,p$, which makes $\E[b_k] = 1/p$. We do not impose the sum-to-one constraint as it may not reflect physical reality. For instance, if the scene contains materials not captured by the chosen endmember set, the abundances of the identified classes need not sum to one. 
    
    \item \textbf{Variance Scaling.} We employ Inverse-Gamma priors to the scaling factors with the constraint that their prior expectation is one. Specifically, $\tau_y, \tau_j \sim \text{Inverse-Gamma}(\nu/2, \nu/2)$, where $\nu$ controls the variance of the scaling factor.
    
    \item \textbf{Covariance Structure.} We employ a conjugate Inverse-Wishart$(\Psi, \rho)$ prior distribution to the shared covariance matrix.
    
    \item \textbf{Endmember Class Medians.} We employ a log-normal prior on the median endmember signature $\mu_j$ centered around a prior value, so that $\log(\mu_j) \sim \mathcal{N}(\log(m_j), \sigma^2 I_n)$.
\end{itemize}
In practice, the hyperparameters $\Theta = \{\nu, \Psi, \rho, \sigma^2, m_j\}$ are estimated from the reference pixels $\{x_{ij}\}$ prior to running the MCMC. Specifically, $\log(m_j)$ is set to the sample mean of $\log(x_{ij})$ across $i$, $\Psi$ is initialized to the sample covariance of the log-transformed reference pixels pooled across classes, and $\sigma^2$ is set to a large value to reflect diffuse prior uncertainty on $\mu_j$. The degrees of freedom $\rho = n + 1$ is chosen as the minimum value for which the Inverse-Wishart prior is proper, yielding a diffuse prior on $\Sigma$. The hyperparameter $\nu = 3$ is chosen to give the scaling factors a prior mean of one while retaining heavy tails, allowing the data to pull $\tau_y$ and $\tau_j$ away from one if warranted.

\subsubsection{Posterior Inference: Parameter estimation and confidence intervals.}
\label{sec:mcmc}
Given the observed mixed spectrum $y$ and reference spectra $\{x_{ij}\}$, we draw approximate samples from the joint posterior distribution $p(b, M, \Sigma, \tau_y, \{\tau_j\} \mid y, \{x_{ij}\})$ using a Markov chain Monte Carlo (MCMC) algorithm. Each iteration produces one sample of all parameters. After running the chain for $T$ iterations and discarding the first $T_0$ as burn-in, point estimates are obtained by averaging the remaining samples, and a $(1-\alpha)$ confidence interval for the $j$-th abundance is given by the empirical $\alpha/2$ and $1-\alpha/2$ quantiles of the posterior samples for $b_j$. The reconstructed spectrum is estimated as the posterior mean of $Mb$.

\begin{algorithm}
\text{MCMC Sampler for the Hierarchical Multiplicative EIV Model}\\\\
\textbf{Input:} Mixed spectrum $y$, reference spectra $\{x_{ij}\}$, hyperparameters $\Theta$, number of iterations $T$. \\
\textbf{Initialize:} $b^{(0)}, M^{(0)}, \Sigma^{(0)}, \tau_y^{(0)}, \{\tau_j^{(0)}\}$ sampled from their respective prior distributions. \\
\textbf{For} $t = 1, \dots, T$, repeat:

\textbf{Metropolis-Hastings Updates:}
\begin{enumerate}
    \item \textbf{Update $b$.} Define $D(b, M) = \text{diag}(Mb)^{-1}$ and the curvature-adapted proposal covariance
    \begin{equation*}
        H(b, M, \tau_y) = \left( M^T D(b, M)^T (\tau_y \Sigma)^{-1} D(b, M) M \right)^{-1}.
    \end{equation*}
    Propose $b^*$ from a multivariate normal distribution centered at $b^{(t-1)}$ with covariance $H(b^{(t-1)}, M^{(t-1)}, \tau_y^{(t-1)})$, truncated to the positive orthant $\mathbb{R}^p_+$ to enforce non-negativity of the abundances. We write this as $b^* \sim \mathcal{TN}_+\left(b^{(t-1)}, H(b^{(t-1)}, M^{(t-1)}, \tau_y^{(t-1)})\right)$. Sampling from this distribution and evaluating its density $q_+$ are performed using the minimax tilting algorithm of \cite{Botev2017}. Since $H$ depends on $b$, the proposal is asymmetric and both directions must be evaluated, set
    \begin{equation*}
        b^{(t)} = \begin{cases} b^* & \text{with probability } \min(1, r), \\ b^{(t-1)} & \text{otherwise,} \end{cases}
    \end{equation*}
    where
    \begin{equation*}
        \log r = \log p(b^* \mid \cdot) - \log p(b^{(t-1)} \mid \cdot) + \log q_+(b^{(t-1)} \mid b^*) - \log q_+(b^* \mid b^{(t-1)}),
    \end{equation*}
    with $p(b \mid \cdot) \propto \exp\left( -\frac{1}{2\tau_y^{(t-1)}}(\log y - \log M^{(t-1)}b)^T (\Sigma^{(t-1)})^{-1} (\log y - \log M^{(t-1)}b) \right) \cdot \exp\left( -p \sum_j b_j \right)$.

    \item \textbf{Update $M$.} For each $j = 1, \dots, p$, define the proposal covariance
    \begin{equation*}
        \Lambda_j = \left( \frac{n_j}{\tau_j^{(t-1)}} (\Sigma^{(t-1)})^{-1} + \frac{1}{\sigma^2} I_n \right)^{-1}.
    \end{equation*}
    Propose $\log (\mu_j^*) \sim \mathcal{N}\left(\log \mu_j^{(t-1)}, \Lambda_j / 5\right)$. The proposal is a symmetric random walk on the log scale, set
    \begin{equation*}
        \mu_j^{(t)} = \begin{cases} \mu_j^* & \text{with probability } \min(1, r), \\ \mu_j^{(t-1)} & \text{otherwise,} \end{cases}
    \end{equation*}
    where $\log r = \log p(\mu_j^* \mid \cdot) - \log p(\mu_j^{(t-1)} \mid \cdot)$ and $p(\mu_j \mid \cdot)$ combines the contributions of $y$ through $\log(M^{(t)}b^{(t)})$ and the reference pixels $\{x_{ij}\}$ for the $j$-th endmember class.
\end{enumerate}

\textbf{Full conditional updates:}
\begin{enumerate}
    \setcounter{enumi}{2}
    \item \textbf{Update $\Sigma$.} Sample from the closed-form full conditional:
    \begin{equation*}
        \Sigma^{(t)} \mid \cdot \sim \text{Inverse-Wishart}\left( \Psi + \frac{r_y r_y^T}{\tau_y^{(t-1)}} + \sum_{j=1}^p \frac{1}{\tau_j^{(t-1)}} \sum_{i=1}^{n_j} r_{ij} r_{ij}^T, \quad \rho + 1 + \sum_{j=1}^p n_j \right),
    \end{equation*}
    where $r_y = \log y - \log(M^{(t)}b^{(t)})$ and $r_{ij} = \log x_{ij} - \log \mu_j^{(t)}$.

    \item \textbf{Update $\tau_j$, $j = 1, \dots, p$.} Sample from the closed-form full conditional:
    \begin{equation*}
        \tau_j^{(t)} \mid \cdot \sim \text{Inverse-Gamma}\left( \frac{n \cdot n_j + \nu}{2}, \quad \frac{\nu + \sum_{i=1}^{n_j} r_{ij}^T (\Sigma^{(t)})^{-1} r_{ij}}{2} \right).
    \end{equation*}

    \item \textbf{Update $\tau_y$.} Sample from the closed-form full conditional:
    \begin{equation*}
        \tau_y^{(t)} \mid \cdot \sim \text{Inverse-Gamma}\left( \frac{n + \nu}{2}, \quad \frac{\nu + r_y^T (\Sigma^{(t)})^{-1} r_y}{2} \right).
    \end{equation*}
\end{enumerate}
\textbf{Output:} After discarding the first $T_0$ samples as burn-in, the remaining $T - T_0$ samples are used as follows. The posterior mean abundance vector is estimated as $\hat{b} = \frac{1}{T - T_0}\sum_{t=T_0+1}^T b^{(t)}$, and the reconstructed spectrum as $\hat{y} = \frac{1}{T-T_0}\sum_{t=T_0+1}^T M^{(t)} b^{(t)}$. A $(1-\alpha)$-confidence interval for the $j$-th abundance is given by the empirical $\alpha/2$ and $1-\alpha/2$ quantiles of $\{b_j^{(t)}\}_{t=T_0+1}^T$.
\end{algorithm}

\section{Experiments}
\label{sec:experiments}

We evaluate the proposed multiplicative model on two hyperspectral imaging datasets. For each dataset, we select two patches of size $30 \times 30$ from the image. The patches are chosen to contain material boundaries and significant contributions from all endmembers.

We compare the multiplicative model against an additive Gaussian model with an identical hierarchical structure, where the mixed spectrum and reference spectra are modeled as
\begin{align*}
    y &= Mb + \varepsilon, & \varepsilon \mid \tau_y, \Sigma &\sim \mathcal{N}(0, \tau_y\Sigma), \\
    x_{ij} &= \mu_j + \varepsilon_{ij}, & \varepsilon_{ij} \mid \tau_j, \Sigma &\sim \mathcal{N}(0, \tau_j\Sigma), \quad i = 1, \dots, n_j.
\end{align*}
Posterior inference for the additive model proceeds via a similar MCMC algorithm, with the log-scale likelihood replaced by its Gaussian counterpart.

Our metrics of comparison are the coverage rates of the 90\% quantile intervals, defined as the 5th and 95th percentile of the posterior samples, the width of those intervals as a function of the abundance, and the signal reconstruction error (SRE). The SRE for a pixel is defined as
\begin{equation*}
    \text{SRE} = 10 \log_{10}\left(\frac{\|y\|_2^2}{\|y - \hat{y}\|_2^2}\right),
\end{equation*}
where $y$ is the observed pixel reflectance and $\hat{y} = \E[Mb|y, \{x_{ij}\}]$, the posterior mean of $Mb$.

For both datasets, abundance maps are available as part of the ground truth. We construct 90\% quantile intervals from the MCMC output for each pixel and endmember. Since posterior samples are strictly positive by construction, the interval lower bound cannot reach zero exactly. For both models, any lower bound below $0.01$ is treated as zero when assessing coverage, reflecting the fact that the posterior is supported on $\mathbb{R}^p_+$ and the 5th percentile of a finite MCMC chain cannot reach zero exactly.

To assess how coverage and interval width vary as a function of the signal strength, we bin the remaining pixels for each endmember into 10 quantile-based bins according to their true abundance values. Each bin contains approximately the same number of pixels, avoiding the sparsely populated bins that would arise from equal-width binning on the skewed abundance distribution.

\subsection{Experiment Results on the Jasper Ridge Dataset}
\label{sec:jasper}

The Jasper Ridge dataset was captured by the Airborne Visible/Infrared Imaging Spectrometer (AVIRIS) sensor \cite{GREEN1998227} operated by the Jet Propulsion Laboratory (JPL) over the Jasper Ridge Biological Preserve in California. The full scene consists of $512 \times 614$ pixels recorded across 224 spectral bands spanning 380--2500 nm at a spectral resolution of 9.46 nm. Following standard preprocessing, a $100 \times 100$ pixel subimage is extracted starting from pixel $(105, 269)$ in the original scene, and bands 1--3, 108--112, 154--166, and 220--224 are removed due to water vapor absorption and atmospheric effects, leaving $n = 198$ channels. The scene contains $p = 4$ endmembers: Road, Soil, Water, and Tree, with ground truth abundance maps provided by \cite{zhu2014spectral}.

\begin{figure}[!htbp]
    \centering
    \includegraphics[width=\columnwidth]{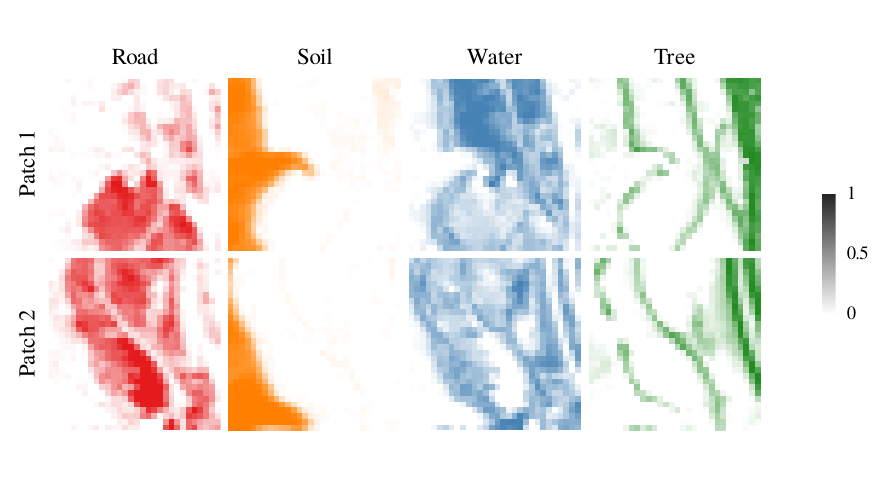}
    \caption{Ground truth abundance maps for both $30 \times 30$ patches extracted from the Jasper Ridge dataset, for each of the four endmember classes: Road, Soil, Water, and Tree.}
    \label{fig:abundance_jasper}
\end{figure}

We extract two $30 \times 30$ patches from the image and identify reference pure pixels by thresholding the ground truth abundances at $0.99$. See Figure~\ref{fig:abundance_jasper} for the ground truth abundance maps of both patches. Figures~\ref{fig:coverage_jasper} and~\ref{fig:sre_jasper} summarize the results. The left column of Figure~\ref{fig:coverage_jasper} shows the conditional coverage rates of the confidence intervals obtained under the multiplicative and additive models for each endmember across the 10 abundance bins. The multiplicative model achieves higher coverage than the additive model in almost every bin, which is the relevant criterion since coverage should hold conditionally at each abundance level rather than on average. The right column shows the mean 90\% quantile interval width as a function of true abundance. While the additive model produces approximately constant interval widths regardless of abundance value, the multiplicative model interval widths scale with the abundance value, consistent with the constant MCV property established in Lemma~\ref{lemma1}.

\begin{figure}[!htbp]
    \centering
    \includegraphics[width=\columnwidth]{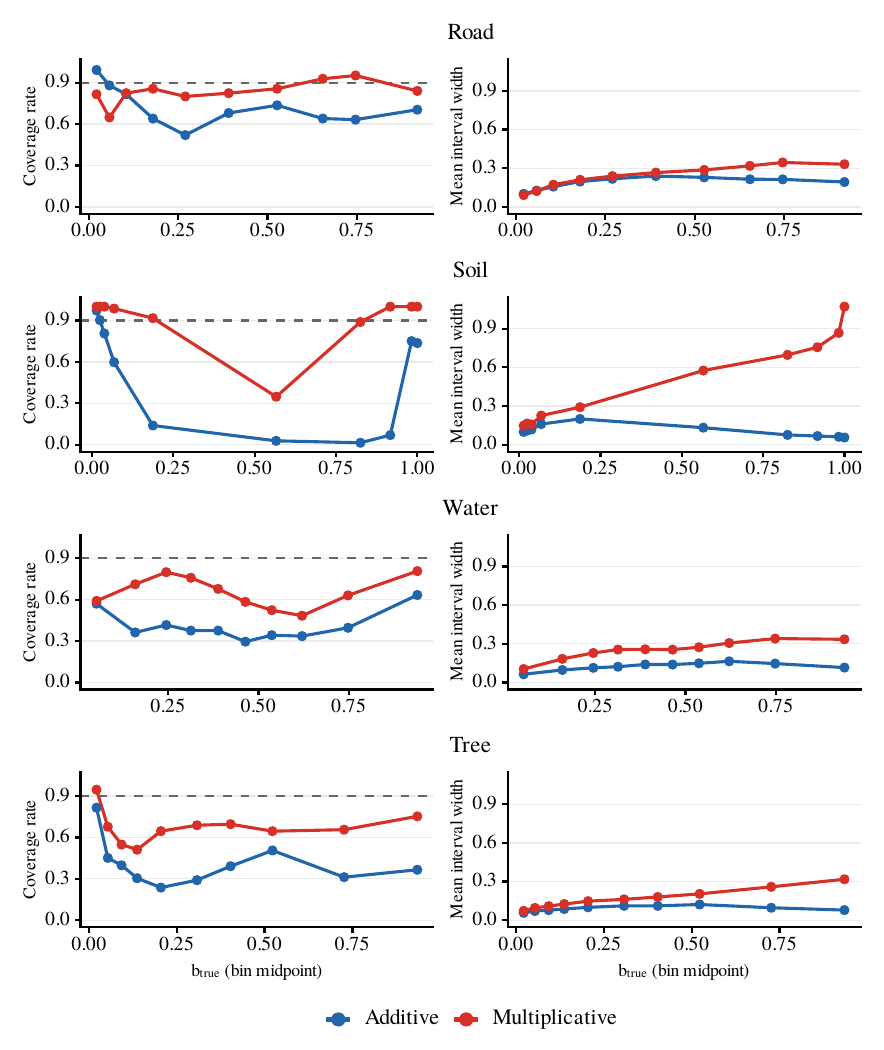}
    \caption{Conditional coverage rates (left) and mean 90\% quantile interval widths (right) as a function of the true abundance value for each endmember on the Jasper Ridge dataset, combined over 1800 pixels from both patches. The dashed line indicates nominal 90\% coverage. Blue corresponds to the additive model and red to the multiplicative model.}
    \label{fig:coverage_jasper}
\end{figure}

Figure~\ref{fig:sre_jasper} plots SRE$_{\text{mult}}$ $-$ SRE$_{\text{add}}$ for each pixel, where red indicates pixels where the multiplicative model reconstructs the signal better and blue the reverse. The difference map is broadly balanced, with neither model consistently dominating. The primary advantage of the multiplicative model lies in its better-calibrated uncertainty quantification, as evidenced by the coverage results.

\begin{figure}[!htbp]
    \centering
    \includegraphics[width=\columnwidth]{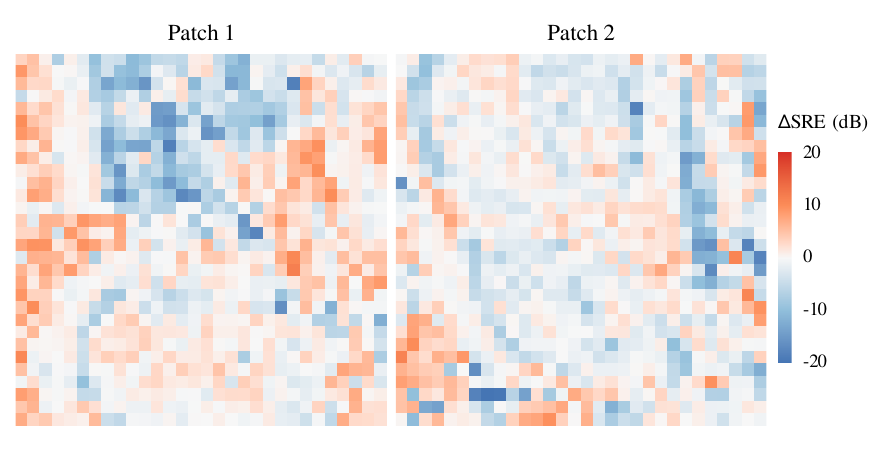}
    \caption{Pixel-wise SRE difference SRE$_{\text{mult}} -$ SRE$_{\text{add}}$ (dB) for both patches on the Jasper Ridge dataset. Red indicates pixels where the multiplicative model achieves better signal reconstruction, blue indicates the reverse.}
    \label{fig:sre_jasper}
\end{figure}

\subsection{Experiment Results on the Urban Dataset}
\label{sec:urban}

The Urban dataset \cite{zhu2014spectral} was acquired by the HYDICE sensor over an urban area at Copperas Cove, TX, USA in October 1995. The scene consists of $307 \times 307$ pixels at a spatial resolution of $2\times 2$ m$^2$, with 210 spectral bands spanning 400--2500 nm at 10 nm spectral resolution. Following standard preprocessing, bands 1--4, 76, 87, 101--111, 136--153, and 198--210 are removed due to water vapor absorption and atmospheric effects, and the remaining bands are further truncated to the first 80, as the reflectance values of certain endmember classes drop to near zero at higher band indices, which we attribute to sensor artifacts rather than physical signal, leaving $n = 80$ channels.

We adopt $p = 4$ endmember classes derived from the 6-endmember ground truth provided with the dataset \cite{zhu2014spectral}. Specifically, we merge Asphalt and Metal into a single class due to their high spectral similarity, and likewise merge Grass and Tree. The spectral angle mapper (SAM) distance between Asphalt and Metal is $6.97^\circ$ with a Pearson correlation of $0.853$, and between Grass and Tree is $14.66^\circ$ with a Pearson correlation of $0.955$, justifying the merging of these pairs. The resulting four classes are: Asphalt (Asphalt $+$ Metal), Grass/Tree, Roof and Dirt. The ground truth abundance for each merged class is obtained by summing the constituent abundances.

We extract two $30 \times 30$ patches and identify reference pure pixels by thresholding the ground truth abundances at $0.99$. See Figure~\ref{fig:abundance_urban} for the ground truth abundance maps of both patches.

\begin{figure}[!htbp]
    \centering
    \includegraphics[width=\columnwidth]{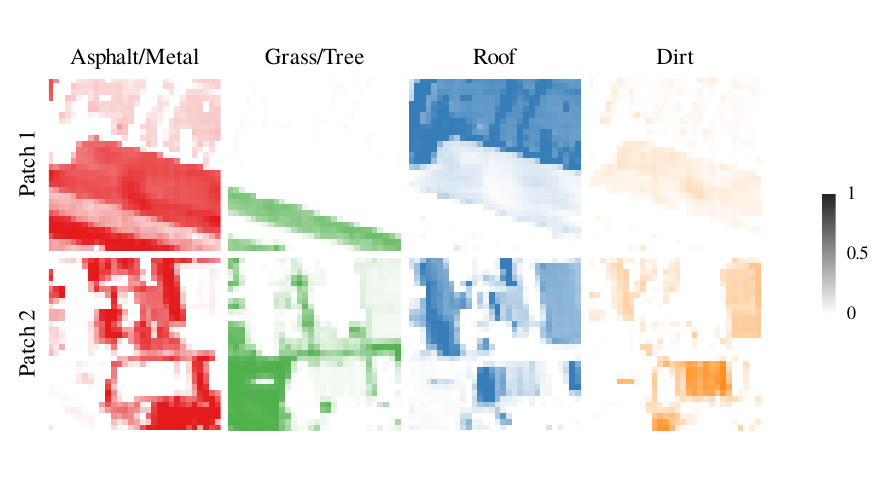}
    \caption{Ground truth abundance maps for both $30 \times 30$ patches extracted from the Urban dataset, for each of the four endmember classes: Asphalt/Metal, Grass/Tree, Roof, and Dirt.}
    \label{fig:abundance_urban}
\end{figure}

The left column of Figure~\ref{fig:coverage_urban} compares the conditional coverage rates for each endmember across the 10 abundance bins. Consistent with the results of Section~\ref{sec:jasper}, the multiplicative model achieves higher coverage than the additive model in almost every bin. The right column shows the interval widths as a function of true abundance, again consistent with the constant MCV property of Lemma~\ref{lemma1}.

\begin{figure}[!htbp]
    \centering
    \includegraphics[width=1\columnwidth]{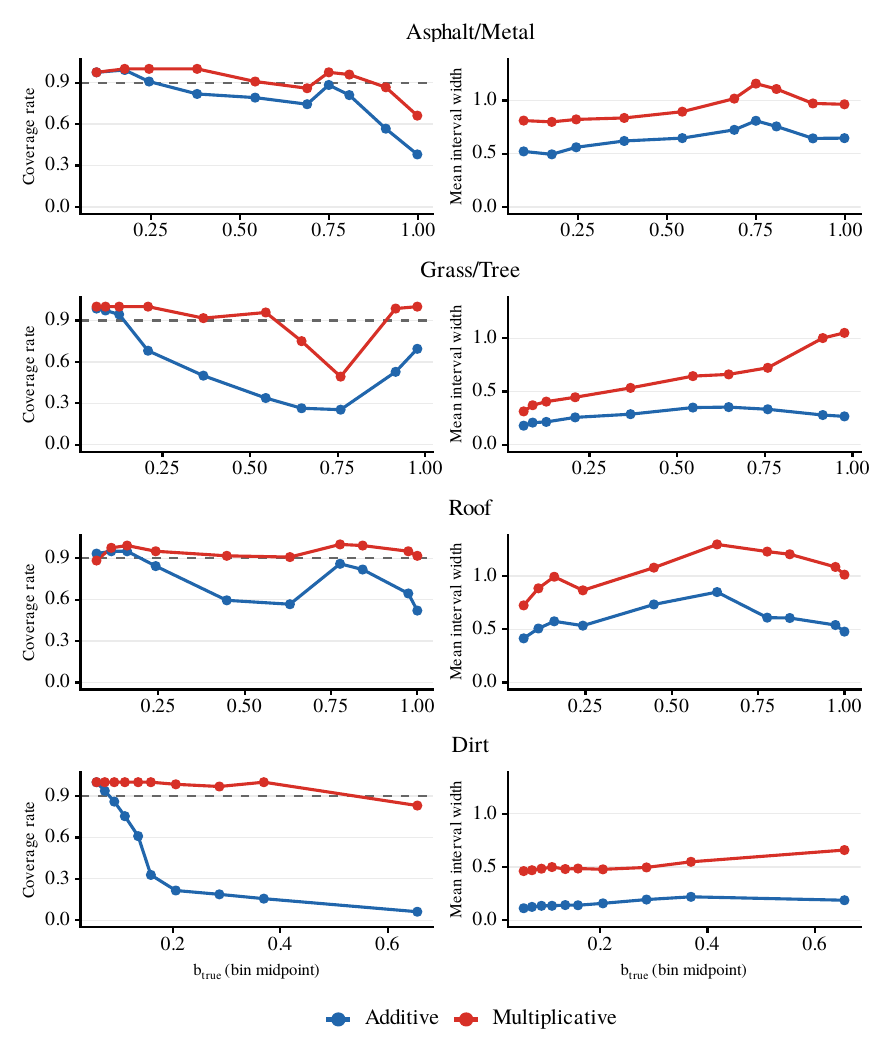}
    \caption{Conditional coverage rates (left) and mean 90\% quantile interval widths (right) as a function of the true abundance value for each endmember on the Urban dataset, combined over 1800 pixels from both patches. The dashed line indicates nominal 90\% coverage. Blue corresponds to the additive model and red to the multiplicative model.}
    \label{fig:coverage_urban}
\end{figure}

Figure~\ref{fig:sre_urban} plots SRE$_{\text{mult}}$ $-$ SRE$_{\text{add}}$ for each pixel, where red indicates pixels for which the multiplicative model reconstructs the signal better and blue is the reverse. The map is predominantly red, suggesting that the multiplicative model is better suited for signal reconstruction on this dataset.

\begin{figure}[!htbp]
    \centering
    \includegraphics[width=\columnwidth]{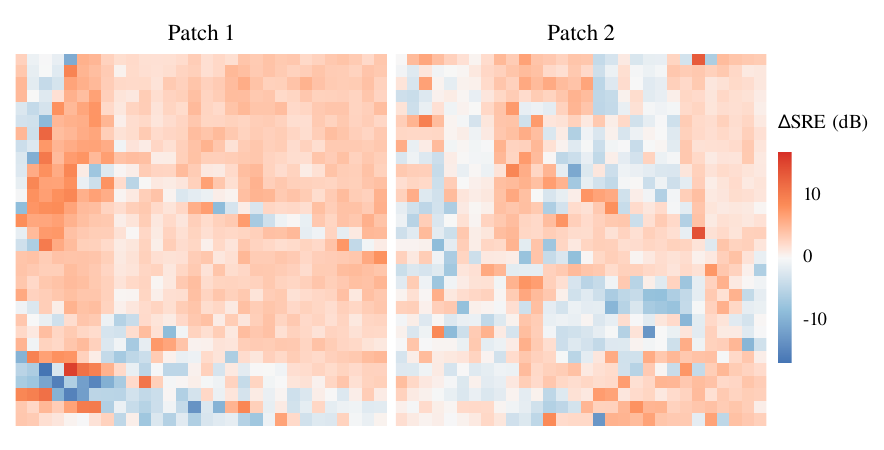}
    \caption{Pixel-wise SRE difference SRE$_{\text{mult}} -$ SRE$_{\text{add}}$ (dB) for both patches on the Urban dataset. Red indicates pixels where the multiplicative model achieves better signal reconstruction, blue indicates the reverse. The predominantly red map suggests the multiplicative model is a better generative model for this scene.}
    \label{fig:sre_urban}
\end{figure}

\section{Discussion}
\label{sec:discussion}

We developed a hierarchical Bayesian multiplicative Errors-in-Variables model for hyperspectral unmixing that accounts for signal-dependent noise, intra-class spectral variability, and inter-band dependence. We establish that the multiplicative model yields abundance confidence intervals that scale proportionally with the abundance value. The experiments on the Urban and Jasper Ridge datasets validate this: the multiplicative model produces better-calibrated coverage rates across abundance levels, with the miscalibration of the additive model being most pronounced at high abundance values. Signal reconstruction quality is comparable between the two models on Jasper Ridge, while the multiplicative model achieves better SRE on the Urban dataset, suggesting it is also a more faithful generative model for scenes dominated by multiplicative noise sources such as topographic shading and atmospheric path radiance.

Several directions remain open for future work. The current model is not well-suited for handling near-zero abundances. When the true abundance of a material is absent or negligible in a pixel, a zero-inflated formulation for $b$ would provide a more principled treatment. The shared covariance matrix $\Sigma$ could be modeled more parsimoniously using a factor model structure, which would reduce the number of parameters and improve scalability to images with a large number of spectral bands. Additionally, while the current model treats the noise as purely multiplicative, real hyperspectral scenes may exhibit a combination of additive and multiplicative perturbations. Incorporating an additive noise component alongside the multiplicative structure would yield a more general model. The current formulation also operates on individual pixels independently. Extending inference to the full image by incorporating spatial structure across pixels represents a natural next step.

While developing our methodology, we also explored a more general formulation in which the endmember contributing to a mixed pixel is treated as a particular realization from the distribution of that class rather than the class median. For example, the tree species present in a given pixel is a random sample from the distribution of trees in the scene, and is not directly observed. This motivates a model where the pixel is a mixture of latent endmember samples $X^* = [x_1^*, \dots, x_p^*]$, where each $x_j^*$ is an unobserved realization from the $j$-th endmember distribution, rather than a mixture of the class medians $\mu_j$:
\begin{align*}
    y &= (X^*b) \odot \varepsilon, & \log(\varepsilon) \mid \tau_y, \Sigma &\sim 
    \mathcal{N}(0, \tau_y\Sigma), \\
    x_{j}^* &= \mu_j \odot \varepsilon_j, & \log(\varepsilon_j) \mid \tau_j, \Sigma 
    &\sim \mathcal{N}(0, \tau_j\Sigma), \\
    x_{ij} &= \mu_j \odot \varepsilon_{ij}, & \log(\varepsilon_{ij}) \mid \tau_j, 
    \Sigma &\sim \mathcal{N}(0, \tau_j\Sigma), \quad i = 1, \dots, n_j.
\end{align*}
While this model is more realistic, it introduces additional latent variables that must be marginalized over during inference. In numerical experiments, we did not observe significant performance advantages in coverage or signal reconstruction over the simpler formulation presented here, and therefore have adopted the latter for our main analysis.

\bibliographystyle{plainnat}
\bibliography{references_main}
\appendix

\section{Verification of the Conditions of \cite{Wu1981}}
\label{app:wu}
 
We verify the conditions of \cite{Wu1981} for the nonlinear least squares problem \eqref{eq:nls}, with regression function $f_i(b) = \log(m_i^T b)$, gradient $f_i'(b) = w_i(b)$, and Hessian $f_i''(b) = -w_i(b) w_i(b)^T$, using normalizing sequence $\tau_n = n$.
 
For consistency, Assumption~A of \cite{Wu1981} requires the identifiability condition $\sum_{i=1}^n (f_i(b) - f_i(b_0))^2 \to \infty$ for every $b \ne b_0$, together with a growth-rate and a Lipschitz condition. Let $c = \inf_{i,\, b\in\Theta} m_i^T b$ and $C_+ = \sup_{i,\, b\in\Theta} m_i^T b$, which are finite and positive by (A1)--(A2). By the mean value theorem applied to $\log$ between $m_i^T b$ and $m_i^T b_0$, there is $\xi_i \in [c, C_+]$ with $f_i(b) - f_i(b_0) = \xi_i^{-1} m_i^T(b - b_0)$, so $(f_i(b) - f_i(b_0))^2 \ge C_+^{-2} (m_i^T(b - b_0))^2$ and hence
\begin{equation*}
    \frac1n \sum_{i=1}^n (f_i(b) - f_i(b_0))^2 \ge \frac{1}{C_+^2}\, (b - b_0)^T \Bigl(\frac1n \sum_{i=1}^n m_i m_i^T\Bigr)(b - b_0).
\end{equation*}
Since $(m_i^T b_0)^2 \ge c^2$, we have $\frac1n \sum_i m_i m_i^T \succeq c^2 \cdot \frac1n \sum_i \frac{m_i m_i^T}{(m_i^T b_0)^2} \to c^2 \bar I(b_0)$, so by monotonicity of the minimum eigenvalue $\liminf_n \lambda_{\min}\bigl(\frac1n \sum_i m_i m_i^T\bigr) \ge c^2 \lambda_{\min}(\bar I(b_0)) > 0$. With $\kappa_n := C_+^{-2} \lambda_{\min}\bigl(\frac1n \sum_i m_i m_i^T\bigr)$,
\begin{equation*}
    \frac1n \sum_{i=1}^n (f_i(b) - f_i(b_0))^2 \ge \kappa_n \|b - b_0\|^2, \qquad \liminf_n \kappa_n \ge \frac{c^2}{C_+^2} \lambda_{\min}(\bar I(b_0)) > 0,
\end{equation*}
which gives both the divergence and the growth-rate condition. The Lipschitz condition holds because (A2) bounds $\|f_i'(b)\| = \|w_i(b)\| \le W$ uniformly. Hence \cite[Theorem 3]{Wu1981} yields strong consistency.
 
For asymptotic normality, Assumption~B of \cite{Wu1981} is met as follows. Condition B(i) is (A3), namely $\frac1n \sum_i f_i'(b_0) f_i'(b_0)^T = \frac1n \sum_i w_i w_i^T \to \bar I(b_0)$ positive definite, with $b_0$ interior by (A1). Condition B(ii) holds since $\|w_i\| \le W$ makes each term $O(n^{-1})$. Condition B(iii) holds because $w_i(b) w_i(b)^T$ is Lipschitz in $b$ uniformly in $i$, giving $\frac1n \sum_i w_i(b) w_i(b)^T \to \bar I(b_0)$ as $b \to b_0$. Conditions B(iv)--B(v) hold because $\|f_i''(b)\| \le W^2$ uniformly and $f_i''$ is likewise uniformly Lipschitz. Since the errors are Gaussian, the least squares estimator is the MLE, and \cite[Theorem 5]{Wu1981} gives the stated limit with variance $\sigma^2 \bar I(b_0)^{-1}$.

\end{document}